\documentclass[aps,pra,groupedaddress,superscriptaddress,twocolumn,twoside]{revtex4}
\usepackage{amsmath,amsfonts,amssymb,amscd,amsthm}
\usepackage{graphicx}
\usepackage{color,comment}
\usepackage{bm}
\usepackage{times,txfonts}
\usepackage[latin1]{inputenc}
\usepackage[toc,page]{appendix}
\usepackage{enumitem}
\usepackage{hyperref}

\newtheorem{theorem}{Theorem}
\newtheorem{lemma}[theorem]{Lemma}

\newtheorem{proposition}[theorem]{Proposition}
\newtheorem{remark}[theorem]{Remark}
\newtheorem{defi/prop}[theorem]{Definition/Proposition}

\newcommand{\N}{\mathbb{N}}

\newcommand{\R}{\mathbb{R}}
\newcommand{\C}{\mathbb{C}}
\newcommand{\EE}{\mathbb{E}}
\renewcommand{\P}{\mathbb{P}}

\renewcommand{\leq}{\leqslant}
\renewcommand{\geq}{\geqslant}

\newcommand{\st}{\  : \ }

\newcommand{\Id}{\mathrm{Id}}

\newcommand{\cD}{\mathcal{D}}

\newcommand{\cS}{\mathcal{S}}

\DeclareMathOperator{\tr}{Tr}

\newcommand{\ketbra}[2]{| #1 \rangle \langle #2 |}
\newcommand{\bra}[1]{\langle #1 |}
\newcommand{\ket}[1]{| #1 \rangle}

\begin{document}

\title{Should Entanglement Measures be Monogamous or Faithful?}

\author{C\'{e}cilia Lancien}
\affiliation{F\'isica Te\`orica: Informaci\'o i Fen\`omens Qu\`antics, Universitat Aut\`onoma de Barcelona, ES-08193 Bellaterra (Barcelona), Spain}
\affiliation{$\mbox{Institut Camille Jordan, Universit\'e Claude Bernard Lyon 1,
43 Boulevard du 11 Novembre 1918, 69622 Villeurbanne Cedex, France}$}

\author{Sara Di Martino}
\affiliation{F\'isica Te\`orica: Informaci\'o i Fen\`omens Qu\`antics, Universitat Aut\`onoma de Barcelona, ES-08193 Bellaterra (Barcelona), Spain}

\author{Marcus Huber}

\affiliation{Group of Applied Physics, University of Geneva, 1211 Geneva 4, Switzerland}
\affiliation{F\'isica Te\`orica: Informaci\'o i Fen\`omens Qu\`antics, Universitat Aut\`onoma de Barcelona, ES-08193 Bellaterra (Barcelona), Spain}
\affiliation{Institute  for  Quantum  Optics  and  Quantum Information (IQOQI), Austrian Academy of Sciences,  Boltzmanngasse 3,  A-1090 Vienna, Austria}

\author{Marco Piani}
\affiliation{SUPA and Department of Physics, University of Strathclyde, Glasgow G4 0NG, United Kingdom}

\author{Gerardo Adesso}
\affiliation{$\mbox{School of Mathematical Sciences, The University of Nottingham, University Park, Nottingham NG7 2RD, United Kingdom}$}

\author{Andreas Winter}
\affiliation{F\'isica Te\`orica: Informaci\'o i Fen\`omens Qu\`antics, Universitat Aut\`onoma de Barcelona, ES-08193 Bellaterra (Barcelona), Spain}
\affiliation{ICREA -- Instituci\'o Catalana de Recerca i Estudis Avan\c{c}ats, Pg. Lluis Companys 23, ES-08010 Barcelona, Spain}

\begin{abstract}
``Is entanglement monogamous?'' asks the title of a popular article [B.~Terhal, IBM~J.~Res.~Dev.~48,~71~(2004)], celebrating C.~H.~Bennett's legacy on quantum information theory. While the answer is affirmative in the qualitative sense, the situation is less clear if monogamy is intended as a quantitative limitation on the distribution of  bipartite entanglement in a multipartite system, given some particular measure of entanglement.
Here, we formalize what it takes for a bipartite measure of entanglement to obey a general quantitative monogamy relation on all quantum states. We then prove that an important class of entanglement measures fail to be monogamous in this general sense of the term, with monogamy violations becoming generic with increasing dimension. In particular, we show that every additive and suitably normalized entanglement measure cannot satisfy any nontrivial general monogamy relation while at the same time faithfully capturing the geometric entanglement structure of the fully antisymmetric state in arbitrary dimension. Nevertheless, monogamy of such entanglement measures can be recovered if one allows for dimension-dependent relations, as we show explicitly with relevant examples.
\end{abstract}

\date{1 August 2016}
\maketitle

{\bf Introduction.} Entanglement is a quintessential manifestation of quantum mechanics \cite{schr,schr2}. The study of entanglement and its distribution reveals fundamental insights into the nature of quantum correlations \cite{horodecki_2009}, on the properties of many-body systems~\cite{reviewarealaw,amico_2008}, and on possibilities and limitations for quantum-enhanced technologies \cite{dowling_2003}. A particularly interesting feature of entanglement is known as {\it monogamy} \cite{terhal_2004}, that is, the impossibility of sharing entanglement unconditionally across many subsystems of a composite quantum system.

In the clearest manifestation of monogamy,
if two parties $A$ and $B$ with the same (finite) Hilbert space dimension are maximally entangled, then their state is a pure state $\ket{\Phi}_{AB}$ \cite{aremax}, and neither of them can share any correlation---let alone entanglement---with a third party $C$, as the only physically allowed pure states of the tripartite system $ABC$ are product states $\ket{\Phi}_{AB} \otimes \ket{\Psi}_C$.
%
Consider now the more realistic case of $A$ and $B$ being in a mixed, partially entangled state $\rho_{AB}$. It is then conceivable for more parties to get a share of such entanglement. Namely, a state $\rho_{AB}$ on a Hilbert space ${\cal H}_A \otimes {\cal H}_B$ is termed ``$n$-shareable''  with respect to subsystem $B$ if it admits a symmetric $n$-extension, i.e.~a state $\rho'_{AB_1\ldots B_n}$ on ${\cal H}_A \otimes {\cal H}_{B}^{\otimes n}$  
invariant under permutations of the subsystems $B_1,\ldots,B_n$ and such that the marginal state of $A$ and any $B_j$ amounts to $\rho_{AB}$.
While even an entangled state can be shareable up to some number of extensions, a seminal result is that a state $\rho_{AB}$ is $n$-shareable {\em for  all}  $n \geq 2$ if and only if it is separable, that is, no entangled state can be infinitely-shareable \cite{fannes_1988,werner_1989,doherty_2004,terhal_2004,yang_2006}. This statement formalizes exactly the monogamy of entanglement (in an asymptotic setting), and has many important implications, including the equivalence between asymptotic quantum cloning and state estimation  \cite{bae_2006, chiribellabecomeclassical}, the emergence of objectivity in the quantum-to-classical transition~\cite{brandao2015generic}, the security of quantum key distribution \cite{ekert_1991,devetak_2005,pawlovski_2010,vazirani_2014}, and the study of frustration and topological phases in many-body systems~\cite{osborne_2006,ferraro_2007,giampaolo_2011,giampaolo_2015,jens_2016}.

Over the last two decades, the goal to formalize monogamy of entanglement in precise quantitative terms and for a finite number of parties attracted increasing interest. The concept of monogamy became synonymous with the validity of an inequality due to Coffman, Kundu and Wootters (CKW) \cite{coffman_2000}. Given any tripartite state $\rho_{ABC}$, and choosing a bipartite entanglement measure $E$, the CKW inequality reads \cite{footnotenote}
\begin{equation}\label{eq:ckw}
E_{A:BC}(\rho_{ABC}) \geq E_{A:B}(\rho_{AB}) + E_{A:C}(\rho_{AC})\,,
\end{equation}
with $\rho_{AB}=\tr_C[\rho_{ABC}]$ and $\rho_{AC}=\tr_B[\rho_{ABC}]$.
Intuitively, Eq.~(\ref{eq:ckw}) means that the sum of the individual pairwise entanglements between $A$ and each of the other parties $B$ or $C$ cannot exceed the entanglement between $A$ and the remaining parties grouped together. Eq.~(\ref{eq:ckw}) was originally proven for arbitrary states of three qubits, adopting the squared concurrence as entanglement measure \cite{coffman_2000}. Variations of the CKW inequality and generalizations to $n$ parties have been established for a number of entanglement measures in discrete as well as continuous variable systems \cite{Dennison,koashi_2004,osborne_2006,adesso_2006,adesso_2007,adesso_2012,kim_2009,kim_2010,fanchini_2013,fanchini_2014,bai_2014,luo_2015,cornelio_2013,regula_2014}. In particular, the squashed entanglement \cite{christandl_2004} and the one-way distillable entanglement fulfill Eq.~(\ref{eq:ckw}) in composite systems of arbitrary dimension \cite{koashi_2004}. Hybrid CKW-like inequalities involving entanglement and other forms of correlations have also been proven \cite{koashi_2004,horodecki2007quantum}, while measures of quantum correlations weaker than entanglement  generally violate the CKW inequality \cite{streltsov_2012}. To some extent, therefore, Eq.~(\ref{eq:ckw}) does capture the spirit of monogamy as a distinctive property of entanglement.

\begin{figure*}[t]
\centering
\begin{minipage}[b]{.28\textwidth}
\vspace*{-1cm}
\includegraphics[width=4.5cm]{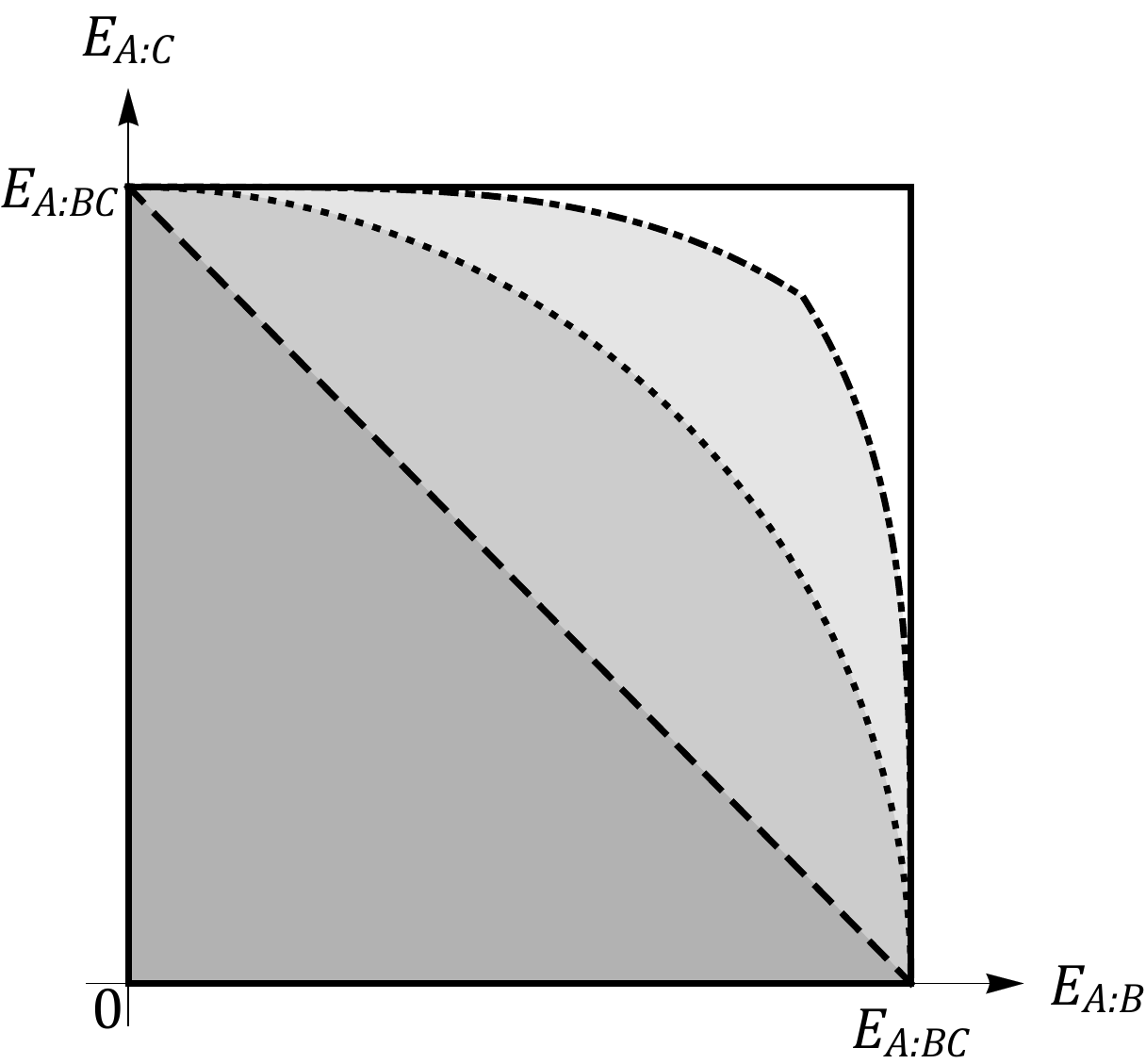}
\end{minipage}
\begin{minipage}[b]{.7\textwidth}
\caption{For any tripartite state $\rho_{ABC}$, an entanglement measure $E$ obeys monogamy if, given the global entanglement $E_{A:BC}$, the pairwise terms $E_{A:B}$ and $E_{A:C}$ are non-trivially constrained. We formalize these constraints via a function $f(E_{A:B}, E_{A:C})$ in Eq.~(\ref{eq:monogamy-relation}). 
Choosing $f(x,y)=x+y$, one gets the CKW inequality (\ref{eq:ckw}), which limits the pairwise terms to the triangular darker region with dashed boundary. The other depicted shaded regions correspond to $f(x,y)=\sqrt{x^2+y^2}$ (dotted boundary) and $f(x,y) = \max(x+c y^4,y+c x^4)$ with $c$ a constant (dot-dashed boundary). Any measure $E$ is termed {\it monogamous} if, for all tripartite states, the ensuing entanglement distribution can be confined to a region strictly smaller than the white square with solid boundary.
The latter denotes the trivial choice $f(x,y)=\max(x,y)$, which is satisfied a priori by any entanglement monotone $E$.}
\label{fig:func}
\end{minipage}
\end{figure*}

The main problem with CKW inequalities, however, is that their validity is not universal, but rather depends on the specific choice of  $E$. Perhaps counterintuitively, several prominent entanglement monotones, such as the entanglement of formation or the distillable entanglement \cite{bennett_1996,virmani_2007,horodecki_2009}, do not obey the constraint formalized by Eq.~(\ref{eq:ckw}), unless one introduces {\it ad hoc} rescalings (see e.g.~\cite{prabhu_2014}).  Since entanglement as a concept {\it is} monogamous in the $n$-shareability sense \cite{terhal_2004}, one is led to raise the following key question: Should any valid {\it entanglement measure} be monogamous in a CKW-like sense?

In this Letter we address the question in general terms. Given an entanglement measure $E$, we shall say that it is {\it monogamous} if there exists a non-trivial function $f:\R_{\geq 0}\times\R_{\geq 0}\rightarrow\R_{\geq 0}$ such that the generalized monogamy relation
\begin{equation} \label{eq:monogamy-relation}
E_{A:BC}(\rho_{ABC}) \geq f\big(E_{A:B}(\rho_{AB}), E_{A:C}(\rho_{AC})\big)\,,
\end{equation}
is satisfied for {\it any} state $\rho_{ABC}$ on {\it any} tripartite Hilbert space ${\cal H}_A \otimes {\cal H}_B \otimes {\cal H}_C$.
Recalling that $E$ is an entanglement monotone (which in turn implies that it is nonincreasing under partial traces), the function $f$ in Eq.~\eqref{eq:monogamy-relation} is without loss of generality such that $f(x,y)\geq \max(x,y)$. Thus, to give rise to a non-trivial constraint, we need $f(x,y) > \max(x,y)$ for at least some range of values of $x$ and $y$.  One might further impose that $f(x,y)$ is monotonic in both its arguments, but we will not require this here. We will however require that $f$ is continuous in general.

While the CKW form \eqref{eq:ckw} of a monogamy relation (which is recovered for the particular choice $f(x,y)=x+y$) implicitly presumes some kind of additivity of the entanglement measure in question, our general form \eqref{eq:monogamy-relation} transcends this and can be applied to recognize any entanglement measure $E$ as {\it de facto} monogamous, based on the intuition that it should obey \emph{some} trade-off between the values of $E_{A:B}$ and $E_{A:C}$ for a given $E_{A:BC}$, see Fig.~\ref{fig:func}. Oppositely, if the only possible choice in Eq.~\eqref{eq:monogamy-relation} were $f(x,y)=\max(x,y)$, then the measure $E$ would fail monogamy in the most drastic fashion: given a state $\rho_{ABC}$, having $E_{A:BC}>0$ a priori would not imply that $E_{A:B}$ and $E_{A:C}$ have to  constrain each other in the interval $[0,E_{A:BC}]$.

Quite remarkably, we rigorously show in the following that the entanglement of formation $E_F$ \cite{bennett_1996} and the relative entropy of entanglement $E_R$ \cite{vedral_1997}, which are two of the most important entanglement monotones for mixed states \cite{virmani_2007,horodecki_2009}, {\it cannot} satisfy a non-trivial monogamy relation in the sense of Eq.~(\ref{eq:monogamy-relation}), with violations becoming generic \cite{hayden_2006} with increasing Hilbert space dimension.
We further show that a whole class of additive entanglement measures, including the entanglement cost $E_F^\infty$ \cite{bennett_1996,vidal_2002} and the regularized relative entropy of entanglement $E_R^\infty$ \cite{brandao_2008,piani_2009},  also fail monogamy as captured by Eq.~(\ref{eq:monogamy-relation}). The latter result is proven by a constructive argument which exploits the peculiar properties of the maximally antisymmetric state on $\C^n \otimes \C^n$, which is $(n-1)$-shareable yet far from separable \cite{christandl_2012}, and has hence been dubbed the `universal counterexample' in quantum information theory \cite{aaronson_2009}. Specifically, any additive entanglement measure which is geometrically {\it faithful} in the sense of being lower-bounded by a quantity with a sub-polynomial dimensional dependence on the antisymmetric state, cannot be monogamous in general.

Our analysis then reveals that entanglement measures divide into two main categories: monogamous (yet geometrically unfaithful) ones, like the squashed entanglement \cite{christandl_2004,koashi_2004}, and geometrically faithful (yet non-monogamous) ones, like $E_F$, $E_R$, and their regularizations. Finally, we show that this dilemma can be resolved if one relaxes the definition~(\ref{eq:monogamy-relation}) to introduce monogamy relations for any {\it fixed} dimension of ${\cal H}_A \otimes {\cal H}_B \otimes {\cal H}_C$. Explicitly, we prove that  $E_F$ and $E_R^\infty$ are retrievable as monogamous for any finite dimension, by providing dimension-dependent choices of $f$ in Eq.~(\ref{eq:monogamy-relation}), which only reduce to the trivial one in the limit of infinite dimension.

{\bf Result (1) Generic non-monogamy for entanglement of formation and relative entropy of entanglement.}
We begin by defining the measures employed in our analysis \cite{bennett_1996,vedral_1997,virmani_2007,horodecki_2009}. The entanglement of formation $E_F$ is the convex roof extension of the entropy of entanglement, $E_F(\rho_{A:B}) = \underset{{\{p_i, \ket{\psi_i}_{AB}\}}}{\inf} \sum_i p_i S(\tr_B [\ket{\psi_i}\!\bra{\psi_i}_{AB}])$, where $S(\rho) = - \tr [\rho \log \rho]$ is the von Neumann entropy, and $\log \equiv \log_2$. On the other hand, the relative entropy of entanglement $E_R$ quantifies the distance from the set of separable states, $E_R(\rho_{A:B}) = \underset{\sigma_{AB}}{\inf} \ S(\rho_{AB} || \sigma_{AB})$, where the minimization is over all separable $\sigma_{AB}$, and $S(\rho||\sigma) = \tr[\rho \log \rho - \rho \log \sigma]$ is the relative entropy.

The fact that both $E_F$ and $E_R$  violate the CKW inequality~(\ref{eq:ckw}) may be traced to their subadditivity, meaning that there exist states $\rho_{AB}$ and $\sigma_{A'C}$ such that $E(\rho_{A:B} \otimes  \sigma_{A':C}) < E(\rho_{A:B}) + E(\sigma_{A':C})$, with $E$ denoting either $E_F$ \cite{hastingsadditivity} or $E_R$ \cite{vollbrecht_2001}. We now show that these measures fail monogamy even in the general sense of Eq.~(\ref{eq:monogamy-relation}). These results are based on {\it random induced states}, defined as follows. Given $n,s\in\N$, a random mixed state $\rho$ on $\C^n$ is induced by $\C^s$ if $\rho=\tr_{\C^s}\ket{\psi}\bra{\psi}$ for $\ket{\psi}$ a uniformly distributed random pure state on  $\C^n\otimes\C^s$. Note that if  $s\leq n$, this is equivalent to $\rho$ being uniformly distributed on the set of mixed states of rank at most $s$ on $\C^n$. Here we focus on the (balanced) bipartite Hilbert space $\C^d\otimes\C^d$, and aim to determine, given $d,s\in\N$, what is the typical value of $E_F(\rho_{A:B})$ and $E_R(\rho_{A:B})$ for $\rho_{AB}$ a random state on $\C^d\otimes\C^d$ induced by an environment $\C^s$. The answer is as follows,
\begin{align}
&\mbox{$\P\left( \left| E_F(\rho_{A:B}) - \left[\log d - \frac{1}{2\ln 2}\right] \right| \leq t \right) \geq 1- e^{-cd^2t^2/\log^2 d}$}, \label{eq:efrandom} \\
&\mbox{$\P\left( \left| E_R(\rho_{A:B}) - \left[\log\left(\frac{d^2}{s}\right)  + \frac{1}{2\ln 2}\frac{s}{d^2}\right] \right| \leq t \right) \geq 1- e^{-cst^2}$}, \label{eq:errandom}
\end{align}
where $s\leq Cd^2t^2/\log^2 (1/t)$ for any fixed $t>0$ in Eq.~(\ref{eq:efrandom}), and $Cd\log(1/t)/t^2\leq s\leq d^2$ for any fixed $0<t<1$ in Eq.~(\ref{eq:errandom}), while $c,C>0$ denote universal constants in both equations. While Eq.~(\ref{eq:efrandom}) was established in
\cite{hayden_2006} (Theorem V.1), although with a looser dependence on the parameters, Eq.~(\ref{eq:errandom}) is an entirely original result of independent interest. Both proofs are mathematically quite involved, and are relegated to \cite{epaps}. \nocite{ZySo,Levy,Milman,ASW,Pisier}

Importantly, failure of monogamy is then retrieved as a {\it generic} trait of entanglement quantified by these measures. Namely, the main result of this section is that there exist states $\rho_{ABC}^{(x)}$ on Hilbert spaces ${\cal H}_A^{(x)} \otimes {\cal H}_B^{(x)}  \otimes {\cal H}_C^{(x)}$ such that, as $x\rightarrow\infty$,
\begin{align}
\label{eq:nomonoEFER}
E_{A:BC}\big(\rho_{ABC}^{(x)}\big)\leq x,\ \text{while}\ E_{A:B}\big(\rho_{AB}^{(x)}\big) \sim E_{A:C}\big(\rho_{AC}^{(x)}\big) \sim x,
\end{align}
for $E$ denoting either $E_F$ or $E_R$. To sketch the proof \cite{epaps}, set $d=\lfloor 2^x\rfloor$ and
${\cal H}_A^{(x)} \equiv {\cal H}_B^{(x)}  \equiv {\cal H}_C^{(x)} \equiv \C^d$. Next, consider $\rho_{ABC}^{(x)}$ a random state on ${\cal H}_A^{(x)} \otimes {\cal H}_B^{(x)}  \otimes {\cal H}_C^{(x)}$, induced by some  ${\cal H}_E^{(x)} \equiv\C^s$, with  $s\sim\log d$. In that way, $\rho_{AB}^{(x)}$ and $\rho_{AC}^{(x)}$ are random states on $\C^d\otimes\C^d$, induced by some  $\C^s\otimes\C^d$, with $s$ and $d$ satisfying both the conditions for Eqs.~(\ref{eq:efrandom}) and (\ref{eq:errandom}) to apply. We then have:
$E_{\{F,R\}}\big(\rho_{A:BC}^{(x)}\big)\leq \log d \leq x$, while
$E_F\big(\rho_{A:B}^{(x)}\big)$ and $E_F\big(\rho_{A:C}^{(x)}\big)$ are both equal to $\log d-O(1)\sim \log d \sim x$ with probability greater than $1-2e^{-cd^2/\log^2 d}$, and $E_R\big(\rho_{A:B}^{(x)}\big)$ and $E_R\big(\rho_{A:C}^{(x)}\big)$ are both equal to $\log d-\log (\log d)-O(1)\sim \log d \sim x$ with probability greater than $1-2e^{-cd\log d}$. \hfill $\square$

{\bf Result (2) Non-monogamy for a whole class of additive entanglement measures.}
We now show that a class of additive entanglement measures
also fail monogamy in the sense of Eq.~(\ref{eq:monogamy-relation}). A key role in this result is played by the antisymmetric state, defined as follows \cite{vidal_2002,christandl_2012}. Given a subsystem $A$ with (finite-dimensional) Hilbert space ${\cal H}_A$, the (maximally) antisymmetric state $\alpha_{A^n}$ on ${\cal H}_A^{\otimes n}$ is the normalized projector onto the antisymmetric subspace of ${\cal H}_A^{\otimes n}$. A crucial property of $\alpha_{A^n}$  is that its reduced state on any group of $k$ subsystems, for any $0\leq k\leq n$,  is $\alpha_{A^k}$, i.e.~the antisymmetric state on ${\cal H}_A^{\otimes k}$.

We now focus on entanglement monotones $E$ satisfying the following conditions:
(a) Normalization: For any state $\rho_{AB}$ on ${\cal H}_A \otimes {\cal H}_B$, $E_{A:B}(\rho_{AB}) \leq \min(\log d_A, \log d_B)$;
(b) Lower-boundedness on the bipartite antisymmetric state: Denoting by $\alpha_{AA'}$ the antisymmetric state on ${\cal H}_A \otimes {\cal H}_{A'}$ (with $d_A = d_{A'}$), $E_{A:A'}(\alpha_{AA'}) \geq c/(\log d_A)^t$, where $c,t>0$ are universal constants;
(c) Additivity on product states: For any state $\rho_{AB}$ on ${\cal H}_A \otimes {\cal H}_B$, $E_{A^m:B^m}(\rho_{AB}^{\otimes m}) = m\,E_{A:B}(\rho_{AB})$;
(d) Linearity on mixtures of locally orthogonal states: For any $0\leq \lambda\leq 1$, and any states $\rho_{AB},\sigma_{AB}$ on ${\cal H}_A \otimes {\cal H}_B$ such that $\tr[\rho_A\sigma_A]= \tr[\rho_B\sigma_B]=0$, $E_{A:B}(\lambda\, \rho_{AB}+(1-\lambda)\sigma_{AB})= \lambda\, E_{A:B}(\rho_{AB}) +(1-\lambda)E_{A:B}(\sigma_{AB})$.

Important examples of entanglement measures fulfilling the above requirements are the regularized versions of the entanglement of formation (aka entanglement cost) $E_F^{\infty}$ \cite{bennett_1996,vidal_2002} and of the relative entropy of entanglement $E_{R}^{\infty}$ \cite{brandao_2008,piani_2009}. Indeed, condition (c) holds by construction for any regularized entanglement measure, defined as
$E_{A:B}^{\infty}(\rho_{AB}) = \underset{n \rightarrow \infty}{\lim} \frac1n E_{A^n:B^n}(\rho_{AB}^{\otimes n})$.
Furthermore, in the case of $E_F^{\infty}$ and $E_R^{\infty}$, conditions (a) and (d) are inherited as they hold for $E_F$ and $E_R$. Finally, condition (b) can be seen as some kind of {\it faithfulness} (or geometry-preserving) property: given that the antisymmetric state has constant trace distance from the set of separable states, one may wish for an entanglement measure to stay bounded away from $0$ on the antisymmetric state, dimension-independently as well (or with a sub-polynomial dependence). For $E$ being $E_F^{\infty}$ or $E_R^{\infty}$, a condition stronger than (b) in fact holds, namely $E_{A:A'}(\alpha_{AA'}) \geq c$, where $c>0$ is a universal constant \cite{christandl_2012}.

What we show here is that any entanglement measure $E$, obeying properties (a)--(d), cannot satisfy a non-trivial monogamy relation in the sense of Eq.~\eqref{eq:monogamy-relation}. We first establish the following result.
Let $n\in\N$, $d=2^n+1$, and set ${\cal H}_{A_j}\equiv \C^d$ for each $0\leq j\leq 2^{n}$. Assume next that $E$ satisfies conditions (a) and (b). Then, there exists $0\leq k\leq n-1$ such that
\begin{eqnarray}\label{eq:lemma31}
E_{{A}_0:{A}_1\ldots{A}_{2^k}}\left(\alpha_{A^{2^k+1}}\right) &= & E_{{A}_0:{A}_{2^k+1}\ldots{A}_{2^{k+1}}}\left(\alpha_{A^{2^k+1}}\right) \nonumber \\
&\geq & \mbox{$\left(1-\frac{\ln (n^{t+1}/c)}{n}\right)$} E_{{A}_0:{A}_1\ldots{A}_{2^{k+1}}}\left(\alpha_{A^{2^{k+1}+1}}\right).
\end{eqnarray}

\begin{figure}[t]
\includegraphics[width=8.5cm]{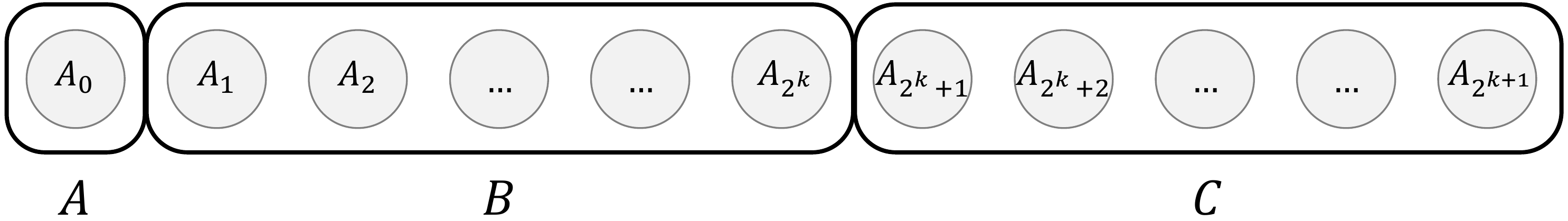}
\caption{Schematic  of the antisymmetric state $\alpha_{A^{2^{k+1}+1}}$ partitioned into three subsystems, $A \equiv A_0$, $B \equiv {A}_1\ldots{A}_{2^k}$, and $C \equiv {A}_{2^k+1}\ldots{A}_{2^{k+1}}$.}
\label{fig:anti}
\end{figure}

To prove Eq.~(\ref{eq:lemma31}), consider a partition of the antisymmetric state $\alpha_{A^{2^{k+1}+1}}$ as illustrated in Fig.~\ref{fig:anti}, and set $g_k = E_{{A}_0:{A}_1\ldots{A}_{2^k}}\left(\alpha_{A_0\ldots A_{2^k}}\right)$ for each
 $0\leq k\leq n$. Then,
\begin{equation} \label{eq:e(k)} \frac{c}{n^t} \approx \frac{c}{(\log d)^t} \leq g_0 \leq g_1 \leq \cdots \leq g_{n-1} \leq g_{n} \leq \log d \approx n. \end{equation}
The last inequality is by property (a), because $d_{A_0}=d = 2^n+1$.
The first inequality is by property (b), because $g_0=E_{{A}_0:{A}_1}\left(\alpha_{A_0A_1}\right)\geq c/(\log d)^t \approx c/n^t$. And the middle inequalities are by monotonicity of $E$ under discarding of subsystems, because for each $0\leq k\leq n-1$, $\alpha_{A_0\ldots A_{2^{k}}}$ is the reduced state of $\alpha_{A_0\ldots A_{2^{k+1}}}$. Now, Eq.~\eqref{eq:e(k)} implies that there exists $0\leq \bar{k}\leq n-1$ such that $g_{\bar{k}}/g_{\bar{k}+1}\geq 1-\ln(n^{t+1}/c)/n$. Indeed, otherwise we would have
$ \frac{g_0}{g_n} = \prod_{k=0}^{n-1} \frac{g_k}{g_{k+1}} < \left(1-\frac{\ln(n^{t+1}/c)}{n}\right)^n \leq \frac{c}{n^{t+1}}$,
which contradicts Eq.~\eqref{eq:e(k)}. Then, as we have on one hand $g_{\bar{k}+1}=E_{{A}_0:{A}_1\ldots{A}_{2^{\bar{k}+1}}}\big(\alpha_{A^{2^{\bar{k}+1}+1}}\big)$, and on the other hand $g_{\bar{k}}= E_{{A}_0:{A}_1\ldots{A}_{2^{\bar{k}}}}\big(\alpha_{A^{2^{\bar{k}}+1}}\big) = E_{{A}_0:{A}_{2^{\bar{k}}+1}\ldots{A}_{2^{\bar{k}+1}}}\big(\alpha_{A^{2^{\bar{k}}+1}}\big)$, Eq.~(\ref{eq:lemma31}) is proven.

The main result of this section then follows immediately. Namely, once again, there exist states $\rho_{ABC}^{(x)}$ on Hilbert spaces ${\cal H}_A^{(x)} \otimes {\cal H}_B^{(x)}  \otimes {\cal H}_C^{(x)}$ such that Eq.~(\ref{eq:nomonoEFER}) holds, for $E$ now denoting any entanglement measure satisfying conditions (a)--(d).

The proof goes as follows. As $E$ satisfies  (a) and (b), we know by Eq.~\eqref{eq:lemma31} that, for any $d\in\N$, there exists a state $\rho_{ABC}$ on ${\cal H}_A \otimes {\cal H}_B \otimes {\cal H}_C$, where ${\cal H}_{A}\equiv\C^d$ and ${\cal H}_{B}\equiv{\cal H}_{C}\equiv(\C^d)^{\otimes 2^k}$ for some $0\leq k\leq \lfloor\log d\rfloor$ (see Fig.~\ref{fig:anti}), such that
$E_{A:B}\left(\rho_{AB}\right)$ and $E_{A:C}\left(\rho_{AC}\right)$ are both lower bounded by $\left(1-\frac{\log(\log^{t+1} d/c)}{\log d}\right)E_{A:BC}\left(\rho_{ABC}\right) = (1-o(1))E_{A:BC}\left(\rho_{ABC}\right)$.
Now, by property (c), for any $m\in\N$, considering $\rho_{ABC}^{\otimes m}$ instead of $\rho_{ABC}$ (and relabelling ${A}^{\otimes m}$ into ${A}$ etc.) will multiply all values of $E$ by a factor $m$. By property (d), for any $0\leq \lambda\leq 1$ and any separable state $\sigma_{ABC}$ (across the cut ${A}:{B}{C}$) which is locally orthogonal to $\rho_{ABC}$, considering $\lambda\,\rho_{ABC}+(1-\lambda)\sigma_{ABC}$ instead of $\rho_{ABC}$ will multiply all values of $E$ by a factor $\lambda$. Consequently, any value $x>0$ for $E_{A:BC}\left(\rho_{ABC}\right)$ is indeed attainable, on some suitably large  Hilbert space ${\cal H}_A^{(x)} \otimes {\cal H}_B^{(x)}  \otimes {\cal H}_C^{(x)}$.
\hfill $\square$

Recapitulating, we demonstrated that entanglement measures which faithfully capture the geometric properties of the antisymmetric state cannot be monogamous in general.
Conversely, there exist relevant entanglement measures for which the desirable condition (b) does not hold---such as the squashed entanglement, which scales as $o(1/d_A)$ on the antisymmetric state $\alpha_{AA'}$---yet monogamy holds instead, even in the original CKW form (\ref{eq:ckw}) \cite{christandl_2004,koashi_2004}. This is the origin of the ``monogamy vs faithfulness'' dilemma discussed in the introduction.

{\bf Result~(3)~Recovering~monogamy:~dimension\!~\!\!\:-\:\!\!~\!dependent relations.}
In the previous two sections, we proved that several important entanglement measures
cannot obey a monogamy relation of the form \eqref{eq:monogamy-relation} with $f$ a universal function. Nevertheless, it may become possible to establish such an inequality if we allow the function $f$ to be dimension-dependent. For instance, the  squared entanglement of formation obeys the CKW inequality for arbitrary three-qubit states \cite{bai_2014}, which means that choosing $f(x,y)=\sqrt{x^2+y^2}$ in Eq.~(\ref{eq:monogamy-relation}), as depicted in Fig.~\ref{fig:func} (dotted boundary), makes $E_F$ monogamous when restricted to Hilbert spaces with $d_A=d_B=d_C=2$.

The third main result of this Letter is to show that non-trivial dimension-dependent monogamy relations can be established for $E_F$ and $E_R^\infty$ in any finite dimension.
Concretely, for any state $\rho_{ABC}$ on a Hilbert space ${\cal H}_A \otimes {\cal H}_B \otimes {\cal H}_C$, it holds
\begin{subequations}\label{eq:edd}
\begin{eqnarray}
\!\!{E_F}\left(\rho_{A:BC}\right) \geq  \max \Big(\!\!\!\!\! & E_F\left(\rho_{A:B}\right) + \frac{c}{d_Ad_C\log^8 d_{A,C}}\big[E_F\left(\rho_{A:C}\right)\big]^8,\ \  \nonumber \\
& E_F\left(\rho_{A:C}\right) + \frac{c}{d_Ad_B\log^8 d_{A,B}}\big[E_F\left(\rho_{A:B}\right)\big]^8 \Big), \label{eq:eddf} \\
\!\! E_R^{\infty}\left(\rho_{A:BC}\right) \geq \max \Big(\!\!\!\!\! & E_R^{\infty}\left(\rho_{A:B}\right) + \frac{c'}{d_Ad_C\log^4d_{A,C}}\big[E_R^{\infty}\left(\rho_{A:C}\right)\big]^4,\nonumber \\
& E_R^{\infty}\left(\rho_{A:C}\right) + \frac{c'}{d_Ad_B\log^4d_{A,B}}\big[E_R^{\infty}\left(\rho_{A:B}\right)\big]^4 \Big),\ \  \label{eq:eddr}
\end{eqnarray}
\end{subequations}
where  $c,c'>0$ are universal constants, and we set $d_{A,B} = \min(d_A,d_B)$, $d_{A,C} =  \min(d_A,d_C)$. An instance of Eq.~(\ref{eq:eddr}) is qualitatively illustrated in Fig.~\ref{fig:func} (dot-dashed boundary).

The proof of Eqs.~(\ref{eq:edd}) makes use of results from Refs.~\cite{piani_2009,brandao_2011,matthews_2009,winter_2015}, and is provided in \cite{epaps}.
While Eqs.~(\ref{eq:edd}) may not be tight \cite{epaps}, they do establish that the involved entanglement measures can be effectively regarded as monogamous according to Eq.~(\ref{eq:monogamy-relation}) in any finite dimension, even though the constraints become trivial in the limit of infinite dimension, in agreement with results (1) and (2). Notice further that Eqs.~(\ref{eq:edd}) encapsulate \emph{strict} monogamy, as the constraining functions satisfy $f(x,y)>\max(x,y)$ for all positive $x$ and $y$. This implies that if, say, $E(\rho_{A:B})=E(\rho_{A:BC})$, then $E(\rho_{A:C})=0$ for $E$ being either $E_F$ or $E_R^\infty$, which means that $A$ and $C$ must be unentangled, as both measures vanish only on separable states.

{\bf Conclusions}. We addressed on general grounds the question of whether entanglement measures should be monogamous in the sense of obeying a quantitative constraint akin to Eq.~\eqref{eq:ckw} introduced in \cite{coffman_2000}. We showed that paradigmatic measures such as the entanglement of formation and the relative entropy of entanglement, as well as their regularizations, cannot be monogamous in general, as they cannot satisfy any non-trivial general relation of the form (\ref{eq:monogamy-relation}) limiting the distribution of bipartite entanglement in arbitrary tripartite states. Monogamy can nonetheless be recovered if the constraints are made dependent on the (finite) dimension of the system.


The present study substantially advances our understanding
of entanglement and the complex laws governing its distribution in systems of multiple parties, and paves the way to more practical developments in quantum communication and computation.
The concept of monogamy as studied here is of particular physical relevance, as the structure of Eq.~(\ref{eq:monogamy-relation}) lends itself to be applied repeatedly to establish limitations in a many-body scenario \cite{amico_2008}, allowing one to compare the distribution of entanglement on equal footing across the various parts of a composite system, unlike e.g.~the case of hybrid monogamy relations involving different quantifiers \cite{koashi_2004}.

It will be worth investigating further links between the
phenomenon of monogamy and the so-called quantum marginal problem \cite{marginalp} as well as  the fact that information cannot be arbitrarily distributed in multipartite quantum states \cite{Lieb,Linden,Cadney,Siewert}.
Implications of our study for progress in other fields like condensed matter \cite{amico_2008} and cosmology \cite{amps_2013}, where monogamy of entanglement takes centre stage, also deserve further study.

{\bf Acknowledgments}. We acknowledge financial support from the European Union's Horizon 2020 Research and Innovation Programme under the Marie Sklodowska-Curie Action OPERACQC (Grant Agreement No.~661338), the European Research Council through the Starting Grant GQCOP (Grant Agreeement No.~637352) and the Advanced Grant IRQUAT (Grant Agreement No.~ 267386), the European Commission through the STREP RAQUEL (Grant Agreement No.~FP7-ICT-2013-C-323970), the Spanish MINECO (Project No.~FIS2013-40627-P), the Generalitat de Catalunya CIRIT (Project No.~2014-SGR-966), the Swiss National Science Foundation (AMBIZIONE PZ00P2{\_}161351), the French CNRS (ANR Projects OSQPI 11-BS01-0008 and Stoq 14-CE25-0033), and the Austrian Science Fund (FWF) through the START Project Y879-N27. G.~A.~thanks B~Regula for fruitful discussions.

\bibliographystyle{apsrev}
\bibliography{nomono}

\begin{thebibliography}{69}
\bibitem[{foo()}]{footnotenote}
\bibinfo{note}{We will denote bipartite entanglement either by
  $E_{A:B}(\rho_{AB})$ or $E(\rho_{A:B})$ throughout the manuscript, depending
  on convenience.}
\end{thebibliography}


\clearpage

\appendix

\widetext
\begin{center} \large{\textbf{Supplemental Material}} \end{center}

\begin{center} \large{\textbf{Should Entanglement Measures be Monogamous or Faithful?}} \end{center}

\section{Entanglement of formation and relative entropy of entanglement of random states}
\label{ap:E_F-E_R}
\setcounter{equation}{0}
Let us first fix a few definitions and notations.
For any $n\in\N$ and $1\leq p\leq \infty$, we denote by $\|\cdot\|_p$ the Schatten $p$-norm on the set of Hermitian operators on $\C^n$. Particular instances of interest are the trace-class norm $\|\cdot\|_1$, the Hilbert--Schmidt norm $\|\cdot\|_2$ and the operator norm $\|\cdot\|_{\infty}$.

Given $n,s\in\N$, we say that a random mixed state $\rho$ on $\C^n$ is induced by $\C^s$, which we denote by $\rho\sim\mu_{n,s}$, if $\rho=\tr_{\C^s}\ket{\psi}\bra{\psi}$ for $\ket{\psi}$ a uniformly distributed random pure state on the global ``system + ancilla'' space $\C^n\otimes\C^s$. It is known that in the case $s\leq n$, $\rho\sim\mu_{n,s}$ is equivalent to $\rho$ being uniformly distributed (for the measure induced by the Hilbert-Schmidt distance) on the set of mixed states of rank at most $s$ on $\C^n$ (see e.g. \cite{ZySo}). In the sequel, we shall use the following shorthand notation: for any pure state $\psi$ on $\C^n\otimes\C^s$, $\rho_{\psi}$ stands for the corresponding marginal state on $\C^n$, i.e. $\rho_{\psi}=\tr_{\C^s}\ket{\psi}\bra{\psi}$.

In the following we will always work on the bipartite Hilbert space $\C^d\otimes\C^d$, where we denote by $\cD(\C^d\otimes\C^d)$ the set of all states and by $\cS(\C^d:\C^d)$ the set of states which are separable across the cut $\C^d:\C^d$. For any $\rho\in\cD(\C^d\otimes\C^d)$, its entanglement of formation is defined as
\[ E_F(\rho) = \inf\left\{ \sum_i p_iS(\rho_{\varphi_i}) \st \rho=\sum_i p_i\ketbra{\varphi_i}{\varphi_i},\ p_i\geq 0,\ \sum_ip_i=1,\ \ket{\varphi_i}\in \C^d\otimes\C^d \right\}, \]
and its relative entropy of entanglement as
\[ E_R(\rho) = \inf\left\{ S(\rho\|\sigma)\st \sigma\in\cS(\C^d:\C^d)\right\} = \inf\left\{ \tr\left(\rho\big(\log\rho-\log\sigma\big)\right) \st \sigma\in\cS(\C^d:\C^d) \right\}.\]
We are interested in determining, given $d,s\in\N$, what is the typical value of $E_F(\rho)$ and $E_R(\rho)$ for $\rho\sim\mu_{d^2,s}$. The answer is stated in Theorems \ref{th:main'} and \ref{th:main} below. They both rely on technical lemmas appearing in Section \ref{sec:technical}. Theorem \ref{th:main'} is then obtained as a direct corollary of Proposition \ref{prop:subspace}. While Theorem \ref{th:main} follows from the upper bound of Proposition \ref{prop:upper bound} and the lower bound of Proposition \ref{prop:lower bound}. Let us also emphasize that the statement about the typical value of $E_F$, appearing in Theorem \ref{th:main'}, was already more or less established in \cite[Thm.~V.1]{hayden_2006}, but only the lower bound appears there and with a not as tight dependence in the parameters, which is why we briefly repeat the whole argument here. Oppositely, the statement about the typical value of $E_R$, appearing in Theorem \ref{th:main}, is new.

\begin{theorem} \label{th:main'}
Fix $t>0$. Let $\rho$ be a random state on $\C^d\otimes\C^d$ induced by some environment $\C^s$, with $s\leq Cd^2t^2/\log^2 d$ for some universal constant $C>0$. Then,
\begin{equation} \label{eq:main'} \P\left( \left| E_F(\rho) - \left[\log d - \frac{1}{2\ln 2}\right] \right| \leq t \right) \geq 1- e^{-cd^2t^2/\log^2 d}, \end{equation}
where $c>0$ is a universal constant.
\end{theorem}

\begin{theorem} \label{th:main}
Fix $0<t<1$. Let $\rho$ be a random state on $\C^d\otimes\C^d$ induced by some environment $\C^s$, with $Cd\log(1/t)/t^2\leq s\leq d^2$ for some universal constant $C>0$. Then,
\begin{equation} \label{eq:main} \P\left( \left| E_R(\rho) - \left[2\log d - \log s + \frac{1}{2\ln 2}\frac{s}{d^2}\right] \right| \leq t \right) \geq 1- e^{-cst^2}, \end{equation}
where $c>0$ is a universal constant.
\end{theorem}

\subsection{A few technical lemmas}
\label{sec:technical}

The two crucial tools that we will use in order to obtain Theorems \ref{th:main'} and \ref{th:main} are Levy's Lemma and Dvoretzky's Theorem (the former being at the heart of the derivation of the latter). These guarantee that, in high dimension, regular enough functions typically do not deviate much from their average behaviour. The precise version of this general paradigm that we will need are quoted in the following Lemmas \ref{lemma:levy} and \ref{lemma:dvo}.

\begin{lemma}[Levy's Lemma for Lipschitz functions on the sphere \cite{Levy}] \label{lemma:levy}
Let $n\in\N$. For any $L$-Lipschitz function $f:S_{\C^n}\rightarrow\R$ and any $t>0$, if $\psi$ is uniformly distributed on $S_{\C^n}$, then
\[ \P( | f(\psi) - \EE f| > t ) \leq e^{-cn t^2/L^2}, \]
where $c>0$ is a universal constant.
\end{lemma}

\begin{lemma}[Dvoretzky's Theorem for Lipschitz functions on the sphere \cite{Milman,ASW}] \label{lemma:dvo}
Let $n\in\N$. For any circled $L$-Lipschitz function $f:S_{\C^n}\rightarrow\R$ and any $t>0$, if $H$ is a uniformly distributed $Cnt^2/L^2$-dimensional subspace of $\C^n$, with $C>0$ a universal constant, then
\[ \P\left(\exists\ \psi\in H\cap S_{\C^n}:\ | f(\psi) - \EE f| > t \right) \leq e^{-cn t^2/L^2}, \]
where $c>0$ is a universal constant.
\end{lemma}

For the sake of completeness, we now re-derive (more or less well-known) concentration results for two functions that will pop up later on while establishing Theorems \ref{th:main'} and \ref{th:main}.

\begin{lemma} \label{lemma:S}
Fix $n,s\in\N$ with $s\leq n$ and $t>0$. If $\psi$ is uniformly distributed on $S_{\C^n\otimes\C^s}$, then
\[ \P\left( \left|S(\rho_{\psi}) - \log s + \frac{1}{2\ln 2}\frac{s}{n} \right| > t \right) \leq e^{-cnst^2/\log^2n}, \]
where $c>0$ is a universal constant.
\end{lemma}

\begin{proof}
Define the function $f:\psi\in S_{\C^n\otimes\C^s}\mapsto S\big(\rho_{\psi}\big)\in\R$. We know from \cite[Lemma III.2]{hayden_2006}, that $f$ is $3\log n$-Lipschitz. Besides, we also know from \cite[Lemma II.4]{hayden_2006}, that $f$ has average (w.r.t. the uniform probability measure over $S_{\C^n\otimes\C^s}$)
\[ \EE_{\psi} f(\psi) = \log s -\frac{1}{2\ln 2}\frac{s}{n} + o\left(\frac{s}{n}\right). \]
Having at hand these two estimates, Lemma \ref{lemma:S} is a direct consequence of Levy's Lemma \ref{lemma:levy}.
\end{proof}

\begin{lemma} \label{lemma:S'}
Fix $d\in\N$ and $t>0$. If $H$ is a uniformly distributed $Cd^2t^2/\log^2 d$-dimensional subspace of $\C^d\otimes\C^d$, with $C>0$ a universal constant, then
\[ \P\left(\exists\ \psi\in H\cap S_{\C^d\otimes\C^d}:\ \left|S(\rho_{\psi}) - \log d + \frac{1}{2\ln 2} \right| > t \right) \leq e^{-cd^2t^2/\log^2d}. \]
\end{lemma}

\begin{proof}
Define the function $f:\psi\in S_{\C^d\otimes\C^d}\mapsto S\big(\rho_{\psi}\big)\in\R$. We know from \cite[Lemma III.2]{hayden_2006}, that $f$ is $3\log d$-Lipschitz. Besides, we also know from \cite[Lemma II.4]{hayden_2006}, that $f$ has average (w.r.t. the uniform probability measure over $S_{\C^d\otimes\C^d}$)
\[ \EE_{\psi} f(\psi) = \log d -\frac{1}{2\ln 2} + o\left(1\right). \]
Having at hand these two estimates, Lemma \ref{lemma:S'} is a direct consequence of Dvoretzky's Theorem \ref{lemma:dvo}.
\end{proof}

\begin{lemma} \label{lemma:norm}
Fix $n,s\in\N$ with $s\leq n$, as well as $x\in S_{\C^n}$, and let $\rho$ be a random state on $\C^n$ induced by an environment $\C^s$. Then,
\[ \forall\ t>0,\ \P\left( \left|\sqrt{\bra{x}\rho\ket{x}}-\frac{1}{\sqrt{n}} \right| > \frac{t}{\sqrt{n}} \right) \leq e^{-cst^2}, \]
where $c>0$ is a universal constant.
\end{lemma}

\begin{proof}
Define the function $g:\psi\in S_{\C^n\otimes\C^s}\mapsto \sqrt{\bra{x}\rho_{\psi}\ket{x}}$. We know from \cite[App.~B]{ASW}, that $g$ is $1$-Lipschitz. Besides, we also know from \cite{ASW}, Appendix B, that $g$ has average (w.r.t. the uniform probability measure over $S_{\C^n\otimes\C^s}$)
\[ \EE_{\psi} g(\psi) = \frac{1}{\sqrt{n}} + o\left(\frac{1}{\sqrt{n}}\right). \]
Having at hand these two estimates, the conclusion of Lemma \ref{lemma:norm} is a direct consequence of Levy's Lemma \ref{lemma:levy}.
\end{proof}

\subsection{Typical value of the entanglement of formation of random induced states}
\label{sec:E_F}

\begin{proposition} \label{prop:subspace}
Fix $t>0$. Let $\rho$ be a random state on $\C^d\otimes\C^d$, induced by some environment $\C^s$, with $s\leq Cd^2t^2\log^2 d$ for some universal constant $C>0$. Then,
\[ \P\left( \exists\ \varphi\in\mathrm{supp}(\rho)\cap S_{\C^d\otimes\C^d}:\ \left|S(\rho_{\varphi})-\log d +\frac{1}{\ln 2}\right| > t \right) \leq e^{-cd^2t^2\log^2 d}, \]
where $c>0$ is a universal constant.
\end{proposition}

\begin{proof} One simply has to observe that $\mathrm{supp}(\rho)$ is a uniformly distributed $s$-dimensional subspace of $\C^d\otimes\C^d$, and apply Lemma \ref{lemma:S'}.
\end{proof}

\subsection{Upper and lower bounds on the typical relative entropy of entanglement of random induced states}
\label{sec:E_R}

Observe the following alternative form of the relative entropy of entanglement: For a state $\rho$ on $\C^d\otimes\C^d$,
\begin{equation} \label{eq:def-E_R} E_R(\rho) = \min\left\{ -\tr(\rho\log\sigma) \st \sigma\in\cS(\C^d:\C^d)\right\} -S(\rho).
\end{equation}

\begin{lemma} \label{lemma:upper bound}
For any state $\rho$ on $\C^d\otimes\C^d$, we have
\[ E_R(\rho) \leq 2\log d - S(\rho). \]
\end{lemma}

\begin{proof}
It is clear from equation \eqref{eq:def-E_R} that we have, in particular,
\[ E_R(\rho) \leq -\tr\left(\rho\log\frac{\Id}{d^2}\right) - S(\rho) = 2\log d - S(\rho). \]
This concludes the proof of Lemma \ref{lemma:upper bound}.
\end{proof}

\begin{proposition} \label{prop:upper bound}
Fix $0<t<1$. Let $\rho$ be a random state on $\C^d\otimes\C^d$ induced by some environment $\C^s$, with $s\leq d^2$. Then,
\[ \P\left( E_R(\rho) > 2\log d - \log s + \frac{1}{2\ln 2}\frac{s}{d^2} +t \right) \leq e^{-cd^2st^2/\log^2 d}, \]
where $c>0$ is a universal constant.
\end{proposition}

\begin{remark}
\label{rem:11}
Note that the bound appearing in Proposition \ref{prop:upper bound} is non-trivial only in the regime $d\leq s\leq d^2$. Indeed, it also holds for any state $\rho$ on $\C^d\otimes\C^d$ that $E_R(\rho)\leq E_F(\rho)\leq \log d$, and if $s<d$, we do not learn anything better than that from Proposition \ref{prop:upper bound} for a random state $\rho$ on $\C^d\otimes\C^d$ induced by some environment $\C^s$. Hence in words, Proposition \ref{prop:upper bound} actually tells us the following: if $\rho\sim\mu_{d^2,s}$ with $d \ll s\leq d^2$, then $E_R(\rho)$ is w.h.p. smaller than $2\log d - \log s + O(s/d^2)$, as $d,s \rightarrow \infty$.
\end{remark}

\begin{proof}
From Lemma \ref{lemma:upper bound}, it is clear that, for any $0<t<1$,
\begin{equation} \label{eq:A} \P\left( E_R(\rho) > 2\log d - \log s + \frac{1}{2\ln 2}\frac{s}{d^2} +t \right) \leq \P\left( S(\rho) < \log s - \frac{1}{2\ln 2}\frac{s}{d^2} - t \right) .\end{equation}
Now, we know from Lemma \ref{lemma:S} that the probability on the r.h.s. of inequality \eqref{eq:A} is smaller than $e^{-cd^2st^2/\log^2d}$, which is precisely the announced result.
\end{proof}

\begin{lemma} \label{lemma:lower bound}
For any state $\rho$ on $\C^d\otimes\C^d$, we have
\[ E_R(\rho) \geq -\log\left(\max\big\{ \bra{x\otimes y}\rho\ket{x\otimes y} \st x,y\in S_{\C^d} \big\}\right) - S(\rho). \]
\end{lemma}

\begin{proof}
We see from equation \eqref{eq:def-E_R} that the only thing we have to prove is that
\begin{equation} \label{eq:B} \min\left\{ -\tr(\rho\log\sigma) \st \sigma\in\cS(\C^d:\C^d)\right\} \geq -\log\left(\max\big\{ \bra{x\otimes y}\rho\ket{x\otimes y} \st x,y\in S_{\C^d} \big\}\right). \end{equation}
Now, for any $\sigma\in\cS(\C^d:\C^d)$, we see by concavity of $\log$ that $-\tr(\rho\log\sigma)\geq -\log\tr(\rho\sigma)$. Indeed, denoting by $\{\lambda_1,\ldots,\lambda_{d^2}\}$ the eigenvalues and by $\{e_1,\ldots,e_{d^2}\}$ an eigenbasis of $\sigma$, we have
\[ \tr(\rho\log\sigma)=\sum_{i=1}^{d^2}\bra{e_i}\rho\ket{e_i}\log\lambda_i\leq \log\left(\sum_{i=1}^{d^2}\bra{e_i}\rho\ket{e_i}\lambda_i\right)=\log\tr(\rho\sigma). \]
As a consequence,
\[ \min\left\{ -\tr(\rho\log\sigma) \st \sigma\in\cS(\C^d:\C^d)\right\} \geq \min\left\{ -\log\tr(\rho\sigma) \st \sigma\in\cS(\C^d:\C^d)\right\} = -\log\left(\max\left\{ \tr(\rho\sigma) \st \sigma\in\cS(\C^d:\C^d)\right\}\right). \]
It then follows from extremality of pure separable states amongst all separable states that
\[ \max\left\{ \tr(\rho\sigma) \st \sigma\in\cS(\C^d:\C^d)\right\} = \max\big\{ \bra{x\otimes y}\rho\ket{x\otimes y} \st x,y\in S_{\C^d} \big\}, \]
so that equation \eqref{eq:B} indeed holds.
\end{proof}

\begin{proposition} \label{prop:lower bound}
Fix $0<t<1$. Let $\rho$ be a random state on $\C^d\otimes\C^d$ induced by some environment $\C^s$, with $C d\log(1/t)/t^2\leq s\leq d^2$ for some universal constant $C>0$. Then,
\[ \P\left( E_R(\rho) < 2\log d - \log s +\frac{1}{2\ln 2}\frac{s}{d^2} - t \right) \leq e^{-cst^2}, \]
where $c>0$ is a universal constant.
\end{proposition}

\begin{remark}
Note that Proposition \ref{prop:lower bound} actually tells us the following: if $\rho\sim\mu_{d^2,s}$ with $d\ll s \ll d^2$, then $E_R(\rho)$ is w.h.p. bigger than $2\log d - \log s - o(1)$, as $d,s\rightarrow+\infty$. And in the case of sufficiently mixed states, the refinement: if $\rho\sim\mu_{d^2,s}$ with $d^{4/3}\ll s\leq d^2$, then $E_R(\rho)$ is w.h.p. bigger than $2\log d - \log s - O(s/d^2)$, as $d,s\rightarrow \infty$. That is, in that regime, the fluctuations
from $2\log d - \log s$ are of the same order $\pm O(s/d^2)$, cf.~Remark~\ref{rem:11}.
We conjecture that for $1\leq s \leq O(d)$, w.h.p.~$E_R(\rho) \geq \log d-O(1)$.
\end{remark}

\begin{proof}
Set $M(\rho)=\max\left\{ \bra{x\otimes y}\rho\ket{x\otimes y} \st x, y\in S_{\C^d}\right\}$. From Lemma \ref{lemma:lower bound}, it is clear that, for any $0<t<1$,
\begin{equation} \label{eq:C} \P\left( E_R(\rho) < 2\log d - \log s + \frac{1}{2\ln 2}\frac{s}{d^2} -t \right) \leq \P\left( -\log M(\rho) < 2\log d -\frac{t}{2} \right) + \P\left( S(\rho) > \log s - \frac{1}{2\ln 2}\frac{s}{d^2} + \frac{t}{2} \right) .\end{equation}
Indeed, defining the events $\mathcal{A}=``\,E_R(\rho)< 2\log d -\log s + 1/(2\ln 2)s/d^2 + t\,''$ and $\mathcal{B}_1=``\,-\log M(\rho) < 2\log d -t/2\,''$, $\mathcal{B}_2=``\,S(\rho) > \log s - 1/(2\ln 2)s/d^2 + t/2\,''$, we have $\neg\mathcal{B}_1\cap \neg\mathcal{B}_2\ \Rightarrow\ \neg\mathcal{A}$, i.e. equivalently $\mathcal{A}\ \Rightarrow\ \mathcal{B}_1\cup \mathcal{B}_2$.
Hence,
\[ \P(\mathcal{A}) \leq \P(\mathcal{B}_1\cup \mathcal{B}_2) \leq \P(\mathcal{B}_1) + \P(\mathcal{B}_2). \]
Now on the one hand, we know from Lemma \ref{lemma:S} that
\begin{equation} \label{eq:D} \P\left( S(\rho) > \log s - \frac{1}{2\ln 2}\frac{s}{d^2} + \frac{t}{2} \right) \leq e^{-cd^2st^2/\log^2d}. \end{equation}
And on the other hand, one may observe that
\[ \P\left( -\log M(\rho) < 2\log d -\frac{t}{2} \right) = \P\left( \sqrt{M(\rho)} > \frac{e^{t/4}}{d} \right) \leq \P\left( \sqrt{M(\rho)} > \frac{1+t/4}{d} \right). \]
So fix $0<\delta<1/8$ and consider $\mathcal{N}_{\delta}$ a $\delta$-net for $\|\cdot\|$ within $S_{\C^d}$. By a standard volumetric argument (see e.g. \cite[Ch.~4]{Pisier}) we know that we can impose $|\mathcal{N}_{\delta}|\leq (3/\delta)^{2d}$. Set next $M_{\delta}(\rho)=\max\left\{\bra{x\otimes y}\rho\ket{x\otimes y} \st x,y\in \mathcal{N}_{\delta}\right\}$. Then, let $x, y\in S_{\C^d}$ and $\bar{x},\bar{ y}\in\mathcal{N}_{\delta}$ be such that $u=x-\bar{x},v= y-\bar{ y}$ satisfy $\|u\|,\|v\|\leq \delta$, and observe that
\[ \bra{x\otimes y}\rho\ket{x\otimes y} = \bra{\bar{x}\otimes\bar{y}}\rho\ket{\bar{x}\otimes\bar{y}} + \bra{\bar{x}\otimes\bar{y}}\rho\ket{\bar{x}\otimes v} + \bra{\bar{x}\otimes v}\rho\ket{\bar{x}\otimes y} + \bra{\bar{x}\otimes y}\rho\ket{u\otimes y} + \bra{u\otimes y}\rho\ket{x\otimes y}. \]
Hence, $\bra{x\otimes y}\rho\ket{x\otimes y} \leq M_{\delta}(\rho) + 4\delta M(\rho)$, so that taking supremum over $x,y\in S_{\C^d}$ yields $M_{\delta}(\rho)\geq (1-4\delta)M(\rho)$. Yet, we know from Lemma \ref{lemma:norm} that, for fixed $x,y\in S_{\C^d}$,
\[\forall\ t>0,\ \P\left( \sqrt{\bra{x\otimes y}\rho\ket{x\otimes y}}> \frac{1+t/8}{d} \right) \leq e^{-c'st^2}. \]
We therefore get by the union bound that
\[ \P\left( \sqrt{M_{\delta}(\rho)} > \frac{1+t/8}{d} \right) \leq \left(\frac{3}{\delta}\right)^{4d}e^{-c'st^2}. \]
Consequently, we eventually obtain, choosing $\delta=t/10$, in order to have $(1+t/4)\sqrt{1-4\delta}\geq 1+t/4-5\delta/4 = 1+t/8$, that
\begin{equation} \label{eq:E} \P\left( \sqrt{M(\rho)} > \frac{1+t/4}{d} \right) \leq \P\left( \sqrt{M_{t/10}(\rho)} > \frac{1+t/8}{d} \right) \leq \left(\frac{30}{t}\right)^{4d}e^{-c'st^2}. \end{equation}
And whenever $s\geq C d\log(1/t)/t^2$, the r.h.s. of inequality \eqref{eq:E} above is smaller than $e^{-c''st^2}$.

So combining the deviation probabilities \eqref{eq:D} and \eqref{eq:E} with inequality \eqref{eq:C}, we get precisely the announced result, just observing that $e^{-c''st^2} + e^{-cd^2st^2/\log^2 d}\leq e^{-c_0st^2}$.
\end{proof}

\section{Dimension-dependent monogamy relations for the entanglement of formation and the regularized relative entropy of entanglement}

\begin{theorem} \label{th:monogamy-dimension}
For any state $\rho_{ABC}$ on $\mathcal{H}_A\otimes\mathcal{H}_B\otimes\mathcal{H}_C$, and any $\eta,\eta'>0$, we have, setting $d_{A,B}=\min(d_A,d_B)$ and $d_{A,C}=\min(d_A,d_C)$,
\[
E_F\left(\rho_{A:BC}\right)
\geq
\max
\left(
	E_F\left(\rho_{A:B}\right) + \frac{c}{d_Ad_C\log^{4+\eta} d_{A,C}}\big[E_F\left(\rho_{A:C}\right)\big]^{4+\eta},
	E_F\left(\rho_{A:C}\right) + \frac{c}{d_Ad_B\log^{4+\eta} d_{A,B}}\big[E_F\left(\rho_{A:B}\right)\big]^{4+\eta}
 \right),
 \]
\[ E_R^{\infty}\left(\rho_{A:BC}\right) \geq \max \left( E_R^{\infty}\left(\rho_{A:B}\right) + \frac{c'}{d_Ad_C\log^{2+\eta'}d_{A,C}}\big[E_R^{\infty}\left(\rho_{A:C}\right)\big]^{2+\eta'},  E_R^{\infty}\left(\rho_{A:C}\right) + \frac{c'}{d_Ad_B\log^{2+\eta'}d_{A,B}}\big[E_R^{\infty}\left(\rho_{A:B}\right)\big]^{2+\eta'} \right), \]
where $c,c'>0$ are constants depending on $\eta,\eta'$.
\end{theorem}

\begin{remark}
For concreteness, in the main text we choose $\eta=4$ and $\eta'=2$.
\end{remark}

To prove Theorem \ref{th:monogamy-dimension}, we will need as a starting point the two monogamy-like relations appearing in Lemma \ref{lemma:monogamy-like} below. In the latter, we denote by $E_{R,\mathrm{LOCC}^{\leftarrow}}$ the relative entropy of entanglement filtered by one-way LOCC measurements, and by $E_{R,\mathrm{LOCC}^{\leftarrow}}^{\infty}$ its regularized version. These were introduced and studied in \cite{piani_2009}, to which the reader is referred for more details.

\begin{lemma} \label{lemma:monogamy-like}
For any state $\rho_{ABC}$ on $\mathcal{H}_A\otimes\mathcal{H}_B\otimes\mathcal{H}_C$, we have
\begin{equation} \label{eq:E_F-E_R-mono-like_app} E_F\left(\rho_{A:BC}\right) \geq E_F\left(\rho_{A:B}\right) + E_{R,\mathrm{LOCC}^{\leftarrow}}\left(\rho_{A:C}\right)\ \text{and}\ E_R^{\infty}\left(\rho_{A:BC}\right) \geq E_R^{\infty}\left(\rho_{A:B}\right) + E_{R,\mathrm{LOCC}^{\leftarrow}}^{\infty}\left(\rho_{A:C}\right). \end{equation}
\end{lemma}

\begin{proof}
The second inequality in \eqref{eq:E_F-E_R-mono-like_app} was proved in \cite[Lemma 2]{brandao_2011}. So let us turn to proving the first inequality in \eqref{eq:E_F-E_R-mono-like_app}. For this, assume that $\{p_i,\psi_{ABC}^{i}\}_i$ is such that $\rho_{ABC}=\sum_ip_i\psi_{ABC}^{i}$ and $E_F\left(\rho_{A:BC}\right)=\sum_i p_i E_F\left(\psi_{A:BC}^{i}\right)$. Now, for a pure state $\psi_{ABC}$, denoting by $\rho_{AB}$ and $\rho_{AC}$ its reduced states on $\mathcal{H}_A\otimes\mathcal{H}_B$ and $\mathcal{H}_A\otimes\mathcal{H}_C$ respectively, we know from \cite[Cor.~2]{koashi_2004},  that $E_F\left(\psi_{A:BC}\right)=E_F\left(\rho_{A:B}\right)+S_{\mathrm{LOCC}^{\leftarrow}}\left(\rho_{A:C}\|\rho_A\otimes\rho_C\right)$. Hence,
\[ E_F\left(\rho_{A:BC}\right)=\sum_ip_i\left[ E_F\left(\rho_{A:B}^{i}\right)+S_{\mathrm{LOCC}^{\leftarrow}}\left(\rho_{A:C}^{i}\|\rho_A^{i}\otimes\rho_C^{i}\right) \right] \geq E_F\left(\rho_{A:B}\right) + S_{\mathrm{LOCC}^{\leftarrow}}\left(\rho_{A:C}\Big\|\sum_ip_i\rho_A^{i}\otimes\rho_C^{i}\right), \]
where the inequality follows from the convexity of $E_F$ and the joint convexity of $S_{\mathrm{LOCC}^{\leftarrow}}$, combined with the observation that $\sum_ip_i\rho_{AB}^{i}=\rho_{AB}$ and $\sum_ip_i\rho_{AC}^{i}=\rho_{AC}$.  Now, the state $\sum_ip_i\rho_A^{i}\otimes\rho_C^{i}$ is obviously separable across the cut $\mathcal{H}_A:\mathcal{H}_C$, and therefore $S_{\mathrm{LOCC}^{\leftarrow}}\left(\rho_{A:C}\big\|\sum_ip_i\rho_A^{i}\otimes\rho_C^{i}\right)\geq E_{R,\mathrm{LOCC}^{\leftarrow}}\left(\rho_{A:C}\right)$, which yields precisely the advertised result.
\end{proof}

\begin{proof}[Proof of Theorem \ref{th:monogamy-dimension}] Since we can exchange the roles played by $\mathcal{H}_B$ and $\mathcal{H}_C$, we just have to show that
\begin{equation} \label{eq:E_F-mono'}  E_F\left(\rho_{A:BC}\right) \geq E_F\left(\rho_{A:B}\right) + \frac{c}{d_Ad_C\log^{4+\eta}\big[\min(d_A,d_C)\big]}\big[E_F\left(\rho_{A:C}\right)\big]^{4+\eta}, \end{equation}
\begin{equation} \label{eq:E_R-mono'} E_R^{\infty}\left(\rho_{A:BC}\right) \geq E_R^{\infty}\left(\rho_{A:B}\right) + \frac{c'}{d_Ad_C\log^{2+\eta'}\big[\min(d_A,d_C)\big]}\big[E_R^{\infty}\left(\rho_{A:C}\right)\big]^{2+\eta'}. \end{equation}
By Pinsker inequality and the Pinsker-type inequality established in \cite[Lemma 3]{brandao_2011}, we know that, for any state $\rho_{AC}$,
\[ E_{R,\mathrm{LOCC}^{\leftarrow}}\left(\rho_{A:C}\right)\geq\frac{1}{2\ln 2}\left\|\rho_{AC} - \mathcal{S}(\mathcal{H}_A:\mathcal{H}_C)\right\|_{\mathrm{LOCC}^{\leftarrow}}^2\ \text{and}\ E_{R,\mathrm{LOCC}^{\leftarrow}}^{\infty}\left(\rho_{A:C}\right)\geq\frac{1}{8\ln 2}\left\|\rho_{AC} - \mathcal{S}(\mathcal{H}_A:\mathcal{H}_C)\right\|_{\mathrm{LOCC}^{\leftarrow}}^2. \]
Now, we also know from \cite{matthews_2009} that, for any states $\rho_{AC}$ and $\sigma_{AC}$,
\[ \|\rho_{AC}-\sigma_{AC}\|_{\mathrm{LOCC}^{\leftarrow}} \geq \frac{c_0}{\sqrt{d_Ad_C}}\|\rho_{AC}-\sigma_{AC}\|_{1}. \]
And finally, it was shown in \cite{winter_2015}, Corollaries 4 and 7, that for any states $\rho_{AC}$ and $\sigma_{AC}$, and any $\eta_0,\eta'_0>0$
\[ | E_F(\rho_{AC}) - E_F(\sigma_{AC}) | \leq C \log\big[\min(d_A,d_C)\big]\|\rho_{AC}-\sigma_{AC}\|_{1}^{1/2-\eta_0}, \]
\[ | E_R^{\infty}(\rho_{AC}) - E_R^{\infty}(\sigma_{AC}) | \leq C' \log\big[\min(d_A,d_C)\big]\|\rho_{AC}-\sigma_{AC}\|_{1}^{1-\eta'_0}. \]
So putting everything together, we get that, for any state $\rho_{AC}$,
\[ E_{R,\mathrm{LOCC}^{\leftarrow}}\left(\rho_{A:C}\right)\geq \frac{c}{d_Ad_C\log^{4+\eta}\big[\min(d_A,d_C)\big]}\big[E_F\left(\rho_{A:C}\right)\big]^{4+\eta}, \]
\[ E_{R,\mathrm{LOCC}^{\leftarrow}}^{\infty}\left(\rho_{A:C}\right)\geq \frac{c'}{d_Ad_C\log^{2+\eta'}\big[\min(d_A,d_C)\big]}\big[E_R^{\infty}\left(\rho_{A:C}\right)\big]^{2+\eta'}. \]
And combining these two lower-bounds with Lemma \ref{lemma:monogamy-like} yields, as wanted, the two inequalities \eqref{eq:E_F-mono'} and \eqref{eq:E_R-mono'}.
\end{proof}


\end{document}